\documentclass[conference]{IEEEtran}
\IEEEoverridecommandlockouts

\usepackage{cite}
\usepackage{lipsum}  \usepackage{mathrsfs}
\usepackage{algorithmic}
\usepackage{graphicx}
\usepackage{textcomp}
\usepackage{xcolor}
\usepackage[left=1.62cm,right=1.62cm,top=1.9cm]{geometry}
\usepackage{amsmath,amssymb,amsfonts,amsthm}
\usepackage{mathtools}
\usepackage{comment}
\usepackage{glossaries}

\usepackage{tikz}

\newcommand\copyrighttext{%
  \footnotesize \textcopyright 2025 IEEE. Personal use of this material is permitted.
  Permission from IEEE must be obtained for all other uses, in any current or future
  media, including reprinting/republishing this material for advertising or promotional
  purposes, creating new collective works, for resale or redistribution to servers or
  lists, or reuse of any copyrighted component of this work in other works.}
\newcommand\copyrightnotice{%
\begin{tikzpicture}[remember picture,overlay]
\node[anchor=south,yshift=10pt] at (current page.south) 
  {\fbox{\parbox{\dimexpr\textwidth-\fboxsep-\fboxrule\relax}{\copyrighttext}}};
\end{tikzpicture}%
}

\usepackage{bm}
\usepackage{array}
\newcounter{thm}
\newtheorem{theorem}[thm]{Theorem}

\newtheorem{definition}[thm]{Definition}

\newtheorem{lemma}[thm]{Lemma}

\newcommand*{\rom}[1]{\uppercase\expandafter{\romannumeral #1\relax}}

\newcommand{\braket}[2]{{|#1\rangle\langle#2|}}

\newcommand{\tr}{{\mathrm{tr}}}

\newcommand{\eins}{{\mathbf{1}}}

\def\BibTeX{{\rm B\kern-.05em{\sc i\kern-.025em b}\kern-.08em T\kern-.1667em\lower.7ex\hbox{E}\kern-.125emX}}
\begin{document}

\title{Achievable Identification Rates in Noisy Bosonic Broadcast Channels
}

\author{
   \IEEEauthorblockN{Zuhra Amiri\IEEEauthorrefmark{1}, Janis N\"otzel\IEEEauthorrefmark{1}}
   \IEEEauthorblockA{
       \IEEEauthorrefmark{1}Emmy-Noether Group Theoretical Quantum Systems Design Lehrstuhl f\"ur Theoretische Informationstechnik,\\ Technische Universit\"at M\"unchen\\ 
  \{zuhra.amiri,  janis.noetzel\}@tum.de 
    }}

\maketitle
\copyrightnotice

\begin{abstract}
Identification in quantum communication enables receivers to verify the presence of a message without decoding its entire content. While identification capacity has been explored for classical and finite-dimensional quantum channels, its behaviour in bosonic systems remains less understood. This work analyses identification over noisy bosonic broadcast channels using coherent states. We derive achievable identification rate regions while ensuring error probabilities remain bounded, even in an infinite-dimensional setting. Our approach leverages quantum hypothesis testing and approximates the infinite sender alphabet with discrete subsets to maintain power constraints.
\end{abstract}


\section{Introduction}

Traditional communication focuses on transmitting messages, whether utilising classical or quantum methods.
 In contrast, identification (ID) serves a different purpose. Rather than focusing on decoding an entire message, identification assesses whether the sender has transmitted a specific message \cite{ahlswededueck, ahlswede2021identification, derebeyouglu2020performance}. This distinction is crucial when the receiver does not need to know the entire message but only verify its presence. Identification is efficient in IoT, sensor networks, and medical or e-health applications to confirm whether a critical event has been detected \cite{idreview}. Furthermore, identification codes can provide an additional layer of security in situations where traditional cryptographic methods might not be available by identifying messages or devices \cite{günlü2021doubly}.

Due to their flexibility and efficiency, identification codes are the backbone of many cutting-edge technologies today. For example, in radiofrequency identification systems, these codes enable smart tags and labels to accurately and reliably pinpoint physical objects \cite{chen2019information,gabsi2021novel}. In cellular networks, these codes facilitate smooth and rapid connections, identifying user devices with minimal communication overhead, ensuring a frictionless user experience \cite{agiwal2018directional}. Moreover, in the domain of digital twin systems, identification codes play a crucial role in maintaining synchronisation between virtual models and their physical counterparts \cite{von2023identification}.

Previous works have established the foundations of classical and quantum identification. Ahlswede and Dueck introduced the concept in the classical setting, while Winter and others started to extended it to quantum channels \cite{ahlswededueck, ahlswede2001strong, loeber1999quantum}. Recent works have also started more practical noise scenarios \cite{identificationOverFading} and investigated deterministic coding strategies \cite{deterministic} as well as realized implementations \cite{designGuide,ferraraImplementation}. The identification capacity is equal to the transmission capacity for single-user channels, though the units differ. In particular, the size of the ID code increases doubly exponentially with block length, assuming the encoder has access to a source of randomness. As a result, identification codes provide an exponential throughput advantage over transmission codes by allowing the encoding and decoding sets to overlap.

This advantage becomes even more significant in the context of bosonic channels used in optical communication. Given the widespread use of bosonic channels in optical communication, understanding their identification capacity is crucial for practical quantum networks.
In modern communication systems, e.g., wireless networks, satellite communications, and multicast services, broadcast channels are pivotal. These channels have a single transmitter that transmits information to two or more receivers. Unlike traditional point-to-point channels, broadcast channels introduce the challenge of efficiently managing resources to cater to the needs of multiple receivers, each with potentially different channel conditions. In \cite{guha2007classical}, the authors studied the transmission capacity of a lossy bosonic broadcast channel. Identification codes have been studied for lossy bosonic broadcast channels in \cite{prd}. Advancing upon this previous work, we model the most common type of noise in the quantum optics setting, which is given by additive thermal noise.
In bosonic quantum systems, identification is particularly challenging. Unlike finite-dimensional quantum channels, bosonic channels operate in infinite-dimensional Hilbert spaces, making identification capacity analysis more complex. Moreover, the study of noisy quantum systems is both practically important and theoretically challenging since the involved states are inevitably given by density matrices and not by vectors anymore. To determine achievable identification rates, we must ensure that error probabilities remain bounded even with an infinite sender alphabet even under such additional challenge.

This work analyses the identification of noisy bosonic channels using coherent states. We derive achievable rate regions and show that identification remains feasible even in an infinite-dimensional setting under power constraints. We approximate the infinite sender alphabet with discrete subsets and apply quantum hypothesis testing to ensure that the error probabilities remain bounded. Our results extend previous findings on quantum identification and provide new insights into the capabilities of bosonic communication systems.

\section{Notation and Channel Model}
We summarise the key mathematical notation and conventions used throughout this work.

We write the identity operator as $\mathbf{1}$ and the trace of an operator $A$ as $\tr(A)$. The trace norm is given as $\|A\|_1 = \tr\sqrt{A^\dagger A}$. The logarithm $\log$ is taken to base two.
We denote the Hilbert space corresponding to a quantum system $A$ as $\mathcal{H}_A$. The set of all density operators $\rho$ on $\mathcal{H}_A$ is denoted as $\mathscr{D}(\mathcal{H}_A)$. A measurement on a quantum system can be defined as a Positive Operator Valued (Probability) Measure (POVM) $\{\Pi_j\}$ satisfying $\Pi_j \geq 0$, $\sum_j\Pi_j = \mathbf{1}$. The quantum entropy of a density operator $\rho$ is defined as $H(\rho) = -\tr[\rho\log(\rho)]$. The quantum entropy of a bipartite state $\sigma_{AB}$ on $\mathcal{H}_A \otimes \mathcal{H}_B$ is defined as $I(A;B)_\sigma = H(\sigma_A) + H(\sigma_B) - H(\sigma_{AB})$. We define the Holevo quantity of an ensemble $\{p_x, \rho_x\}_{x \in \mathcal{X}}$ as $\chi(\{p_x, \rho_x\}_{x \in \mathcal{X}})$. We define the Fock space denoted as $\mathcal F=\mathrm{span}(\{|k\rangle\}_{k\in\mathbb N})$ and for every $L\in\mathbb N$ the finite-dimensional subspace $\mathcal F_L=\mathrm{span}(\{|k\rangle\}_{k=0}^{L-1})$.

We consider a quantum broadcast channel $\mathcal{N}_{A\to B_1 B_2}$ with one sender $A$ and two receivers $B_1$ and $B_2$. Mathematically, it is represented as a completely positive trace-preserving (CPTP) map $\mathcal{N}_{A\to B_1 B_2}: \mathscr{D}(\mathcal{H}_A) \to  \mathscr{D}(\mathcal{H}_{B_1} \otimes \mathcal{H}_{B_2})$ \cite{prd}. The transmission occurs over a noisy bosonic channel with transmissivity $0 \leq\tau\leq1$ and noise levels $N_i$, $i \in \{1,2\}$ at receiver $i$, respectively. The input-output relations for the two receivers can be shown with the beam splitter model
    \begin{align}
        \hat{b}_1 &= \sqrt{\tau}\hat{a} + \sqrt{1-\tau}\hat{e}\\
        \hat{b}_2 &= \sqrt{1-\tau}\hat{a}-\sqrt{\tau}\hat{e},
    \end{align}
where $\hat{e}$ is the environmental noise. 
Furthermore, the channel is memoryless, meaning if the systems $A^n = (A_1, \ldots, A_n)$ are sent through $n$ channel uses, then the input state $\rho_{A^n}$ goes through the tensor product mapping $\mathcal{N}_{A^n\to B_1^n B_2^n} = \mathcal{N}^{\otimes n}_{A\to B_1 B_2}$. The marginal channels are defined as follows: for receiver 1, $\mathcal{N}^{(1)}_{A \to B_1}(\rho_A) = \tr_{B_2}(\mathcal{N}_{A \to B_1, B_2}(\rho_A))$, and similarly, for receiver 2 $\mathcal{N}^{(2)}_{A \to B_2}(\rho_A) = \tr_{B_1}(\mathcal{N}_{A \to B_1, B_2}(\rho_A))$. 

We assume that the sender transmits coherent states 
\begin{align}
    |\alpha\rangle = e^{-|\alpha|^2/2}\sum_{r=0}^\infty\frac{\alpha^r}{\sqrt{r!}}|r\rangle
\end{align}
with $\alpha \in \mathbb{C}$.
Assuming that there is some leakage of the emitted field into the environment and taking into account that the receivers only ever apply individual measurements and never collective ones, the channel can then be written as a classical-quantum channel
\begin{align}
    \mathcal N(\alpha)=S_{N_1}(\sqrt{\tau_1}\alpha)\otimes S_{N_2}(\sqrt{\tau_2}\alpha)
\end{align}
where $S_N(\alpha)$ is a displaced thermal state \cite{Holevo_2019}.
In the given setting, an $(M_1, M_2, n, \lambda_1, \lambda_2)$-ID code $\mathcal{C}_n$ over a quantum broadcast channel with power constraint $E$ consists of a set of input sequences $\{\alpha^n_{m,1},\ldots,\alpha^n_{m,M_1}\}$ where $m=(m_1,m_2)$ satisfying for all $m\in[M_1]\times[M_2]$ the energy constraint  
\begin{align}
\sum_{i=1}^{n} |\alpha_{m,i}|^2 \leq E.
\end{align}  
We define the POVMs
\begin{align}
\mathbf{\Pi}_{B_1^n}^{m_1} = \{\mathbf{1}-\Pi_{m_1}^{(1)}, \Pi_{m_1}^{(1)} \}, \quad  
\mathbf{\Pi}_{B_2^n}^{m_2} = \{\mathbf{1}-\Pi_{m_2}^{(2)}, \Pi_{m_2}^{(2)} \}
\end{align}  
for $ m_1 \in [M_1] $ and $ m_2 \in [M_2] $, respectively. The sender selects a pair of messages $ (m_1, m_2) $ with $ m_i \in [M_i] $, where each receiver is interested in decoding their respective message. 

The identification conditions are:  
\begin{align} 
\min_{1\leq m_1\leq M_1} \tr\left((\Pi^{(1)}_{m_1}\otimes\mathbf{1})\boxplus^\mathbf{N}_\tau(\alpha^n_{(m_1,m_2)})\right) &\geq 1-\lambda_1, \label{eq:cond1} \\ 
\min_{1\leq m_2\leq M_2} \tr\left((\eins\otimes \Pi^{(2)}_{m_2})\boxplus^\mathbf{N}_\tau(\alpha^n_{(m_1,m_2)}\right) &\geq 1-\lambda_2, \label{eq:cond2} \\ 
\max_{m_1\neq m'_1} \tr\left((\Pi^{(1)}_{m'_1}\otimes\eins)\boxplus^\mathbf{N}_\tau(\alpha^n_{(m_1,m_2)}\right) &\leq \lambda_1. \label{eq:cond3}\\ 
\max_{m_2\neq m'_2} \tr\left((\eins\otimes \Pi^{(2)}_{m'_2}) \boxplus^\mathbf{N}_\tau(\alpha^n_{(m_1,m_2)}\right) &\leq \lambda_2\label{eq:cond4}.
\end{align} 

We define two types of identification errors: (1) The missed identification error, denoted as $e_{1,m_i}^{(i)}$, occurs when receiver $i$ fails to detect message $m_i$ when $m_i$ was sent. This is controlled by Eq. \eqref{eq:cond1} and Eq. \eqref{eq:cond2}. (2) The false identification error, denoted as $e_{2,m_i}^{(i)}$, occurs when receiver $i$ incorrectly identifies a message $m'_i \neq m_i$ when $m_i$ was sent. This is controlled by Eq. \eqref{eq:cond3} and Eq. \eqref{eq:cond4}.

A non-negative rate $ R_i $ is achievable if there exists a sequence of such codes $\mathcal{C} = (\mathcal{C}_n)_{n\in\mathbb{N}}$ with parameters $(M_{i,n}, n, \lambda_{i,n})$ satisfying  
\begin{align}
R_i \leq \liminf_{n\to\infty} \frac{1}{n} \log \log M_{i,n}, \quad \text{and} \quad \lim_{n\to\infty} \lambda_{i,n} = 0.
\end{align}  
The supremum over all achievable rates is called the identification capacity. The receiver accepts a message if the measurement outcome $ s_i = 1$ and rejects it if $ s_i = 0$.

\section{Definitions and Lemmas}
In this section, we will show essential definitions. 

\begin{lemma}[Gentle Operator Lemma \cite{watts2024quantum}]\label{lem:gentle}
    Let $\rho$ be a density operator and $0 \leq \Pi \leq \mathbf{1}$ be a measurement operator. If $\tr(\Pi \rho) \geq 1 - \epsilon$ for some $\epsilon \geq 0$, then
\begin{align}
    \left\| \sqrt{\Pi} \rho \sqrt{\Pi} - \rho \right\|_1 \leq 2 \sqrt{\epsilon},
\end{align}

where $\| \cdot \|_1$ denotes the trace norm. 
\end{lemma}

\begin{lemma}[see \cite{bracher2017}]\label{lem:lemma5}
Let $ \mathcal{C}_n$ be a parent codebook of size $ e^{n R_P}$, with $R_P$ being the pool rate. For identification coding, we randomly assign to each message $ m \in \mathcal{M} $ an index set $ \mathcal{V}_m \subset \mathcal{C}_n $ of expected size $e^{n\tilde{R}}$, with $\Tilde{R}$ being the binning rate, with $ R < \tilde{R} < R_P $, with $R$ being the ID rate. Let $ \mu > 0$ be such that
\begin{align}
0 < \mu < \min\{R_P - \tilde{R}, \tilde{R} - R\}.
\end{align}
Then, with high probability over the random assignment:
\begin{align}
|\mathcal{V}_m| &> (1 - \delta_n)e^{n\tilde{R}}   \\
|\mathcal{V}_m| &< (1 + \delta_n)e^{n\tilde{R}}   \\
|\mathcal{V}_m \cap \mathcal{V}_{m'}| &< 2\delta_n e^{n\tilde{R}}  
\end{align}
for all distinct $ m, m' \in \mathcal{M} $, where $ \delta_n = e^{-n\mu/2} $. Moreover, the failure probability of these events vanishes:
\begin{align}
\lim_{n \to \infty} \mathbb{P}\left[\{\mathcal{V}_m\}_{m \in \mathcal{M}} \notin \mathcal{G}_\mu\right] = 0,
\end{align}
where $ \mathcal{G}_\mu$ is the set of "good" assignments satisfying the above three properties.
\end{lemma}

\begin{definition}[Capacity of Noisy Bosonic Channel \cite{Giovannetti_2014}]
    The capacity of a noisy bosonic channel is given by
\begin{equation}
    C = g(\tau E + (1-\tau) N) - g((1-\tau)N),
\end{equation}
where $g(x)$ is the entropy function for a thermal state, $\tau$ the transmissivity, $E$ the average photon number constraint, and $N \in [0,\infty[$ being the additive noise.
\end{definition}

\section{Results}

\begin{theorem}\label{thm:main}
      Let $\mathcal{N}$ be a noisy bosonic broadcast channel with power constraint $E$, transmissivity $\tau_i$, and additive noise levels $N_i$ at each receiver $i$, $i \in \{1,2\}$.
    Then, an achievable identification rate region is given by:
    \begin{align}
\mathcal{R}_{\text{ID}}^{\text{noisy}}&(\mathcal{N}) =\nonumber\\ 
&\left\{
(R_1, R_2) : \begin{array}{ll}
 R_1 \leq g(\tau_1 E + N_1) -  g(N_1)  \\
 R_2 \leq g(\tau_2 E + N_2) - g(N_2)
\end{array}\right\}
\end{align}
\end{theorem}

To prove Theorem \ref{thm:main}, we proceed in three steps: (1) We discretise the sender alphabet $\mathbb{C}$ by selecting a discrete subset $\mathbf{A}$ containing exactly $x\in\mathbb{N}$ states. (2) Before further action, we let each receiver apply a threshold detection for the received energy to each received signal and report an error if the energy exceeds a threshold $L\approx\log n$ for any of the received signals, where $n$ is the number of channel uses. This makes the receiver system effectively finite-dimensional so that known tools from the analysis of finite-dimensional systems apply. This lets us prove the existence of a sequence of achievable rates described by Holevo quantities. (3) In the final step, we show that this sequence of Holevo quantities converges towards the Gordon functions, as stated in Theorem \ref{thm:main}.

The required technical statements are the following
\begin{theorem}\label{thm:approximate-main}
    Let $\mathcal{N}$ be a noisy bosonic broadcast channel with power constraint $E$, transmissivity $\tau_i$, and additive noise levels $N_i$ at each receiver $i$, $i \in \{1,2\}$. Suppose the input alphabet is restricted to a discrete subset $\mathbf{A}$ of coherent states.

     Then an achievable identification rate region over $\mathcal{N}$ is:
     \begin{align}
\mathcal{R}_{\text{ID}, \mathbf{A}}^{\text{noisy}}&(\mathcal{N}) =\nonumber\\ 
&\left\{
(R_1, R_2) : \begin{array}{ll}
 R_1 \leq \chi({p(a), \mathcal{N}_1(a)}{a \in \mathbf{A}})  \\
 R_2 \leq \chi({p(a), \mathcal{N}_2(a)}{a \in \mathbf{A}})
\end{array}\right\}
\end{align}
\end{theorem}

We summarise the construction from the Appendix on Gaussian discretisation of \cite{seitz2025capacity} into the following Lemma:
\begin{lemma}\label{lem:approximation-lemma}[C.f. {\cite{seitz2025capacity}}]
    To every Gaussian ensemble $G(0,E)$ with probability density function $p(\alpha):=\frac{1}{\pi E} e^{-|\alpha|^2/E}$ there exists a sequence $(p_x)_{x\in\mathbb N}$ of probability distributions on $\mathbb C$ having support on respective discrete alphabets $\mathbf A_x$, $x\in\mathbb N$, so that 
    \begin{align}
        |\chi(\{p_i(\alpha),\mathcal N_{\tau,E}(\alpha)\}_{\alpha \in \mathbf{A}_x}) - g(\tau E) + g(\tau E + N)|\leq \epsilon_x
    \end{align}
    for a sequence $(\epsilon_x)_{x\in\mathbb N}$ satisfying $\lim_{x\to\infty}\epsilon_x=0$.
\end{lemma}

\begin{lemma}[See \cite{nötzel2025infinitefoldquantumadvantageclassical}]\label{lem:truncated-projector}
Let $\alpha \in \mathbb{C}$, $L \in \mathbb{N}$, $T_L := \sum_{n=0}^{L} |n\rangle \langle n|$. Then
\begin{align}
  \tr (T_L|\alpha\rangle \langle \alpha|) \geq 1 - \frac{1}{L!} \cdot \max\{2, |\alpha|^2\}^L. 
\end{align}
Let $E > 0$. If $L \geq 2 \cdot 50^2 \cdot E \cdot \log(e^{-1})$, then
\begin{align}
   \tr (T_L \frac{1}{\pi E} \int e^{-|\alpha|^2 / E} |\alpha\rangle \langle \alpha| d\alpha) \geq 1 - 2^{-L}. 
\end{align}
\end{lemma}

We prove Theorem \ref{thm:main} based on Theorem \ref{thm:approximate-main} and Lemma \ref{lem:approximation-lemma} by showing that $\mathcal{R}_{\mathbf A}$ can approximate $\mathcal{R}$ to arbitrary precision given a large enough $\mathbf{A}$. 

The proof of Theorem \ref{thm:approximate-main} is based on Lemma \ref{lem:approximation-lemma}, which lets us use information theory techniques for finite-dimensional and bosonic systems.

We note that Lemma \ref{lem:approximation-lemma} itself uses the prior works \cite{Winter2015} and \cite{Becker2021} in its proof.

To illustrate the effect of the input energy on the identification capabilities of noisy bosonic broadcast channels, we plot the achievable identification rate regions for two different average photon numbers in Figure \ref{fig:plot}. These regions demonstrate how varying energy constraints influence the identification performance of the two receivers.

\begin{figure}
    \centering
    \includegraphics[width=0.9\linewidth]{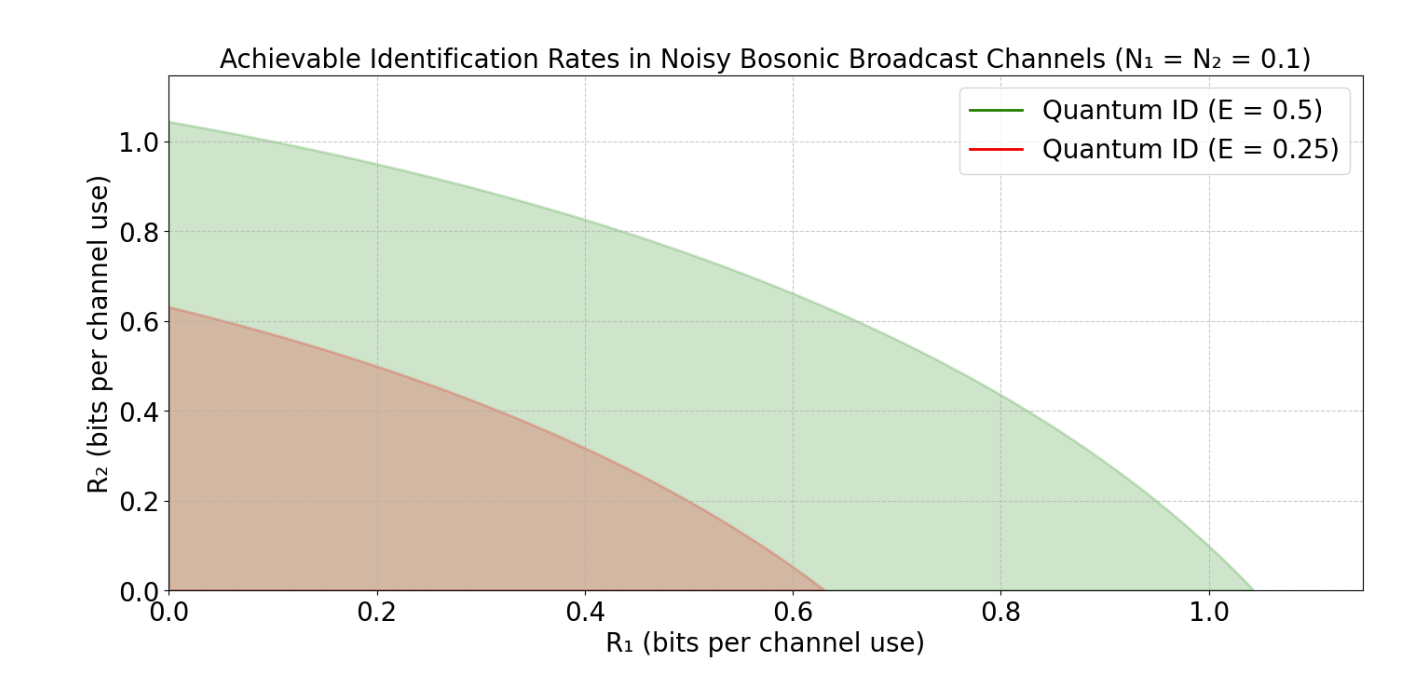}
    \caption{
   A plot of the achievable identification rate regions for a noisy bosonic broadcast channel with additive noise $ N_1 = N_2 = 0.1$, for average input photon numbers $E = 0.5$ (green region) and $E = 0.25$ (red region), as shown in Theorem~\ref{thm:main}.}
    \label{fig:plot}
\end{figure}

\section{Proofs}

\begin{proof}[Proof of Theorem~\ref{thm:approximate-main}]
    We will take the following steps to show that the rate region is achievable. First, we will show the code construction in more detail, then design the decoder, and lastly, we will analyse the error probabilities.

    We pick a discrete set $\mathbf A\subset\mathbb C$ with the property $\max_{\alpha\in\mathbf{A}}|\alpha|^2\leq E'$ for some $E'>0$. A random binning strategy is used to encode the messages \cite{bracher2017}. The sender creates a shared pool of random codewords $\{\alpha^n(v) \}_{v \in \mathcal{V}}\subset\mathbf A^n$, where $\mathcal{V}$ is a large set of codeword indices with $|\mathcal{V}| = e^{nR_P}$ and  $R_P$ is the pool rate. It holds $\max\{\Tilde{R}_1,\Tilde{R}_2\} < R_P < \Tilde{R}_1 + \Tilde{R}_2$, where $\Tilde{R}_i$ is the binning rate for receiver $i$. This condition ensures that there are enough codewords per bin ($R_P > \Tilde{R}_i$) and that there are non-empty intersections ($R_P < \Tilde{R}_1 + \Tilde{R}_2$).
    The binning rate is fixed such that $R_i < \Tilde{R}_i < I(A;B_i)_{\sigma}$ for each receiver $i$. 
    The code size $M_i = \exp(e^{nR_i})$ is double exponential.
    We will first look at the single-mode bosonic channel $\mathcal{N}$ to use Lemma \ref{lem:approximation-lemma}. 
    To save space, we occasionally write $\rho_\alpha$ instead of $S_{N_i}(\sqrt{\tau_i}\alpha)$ for brevity. 
    With Lemma $\ref{lem:truncated-projector}$ we have for large enough $L\in\mathbb N$
    \begin{align}
        \tr(\rho_\alpha T_L) \geq 1- 2^{-L}.
    \end{align}
    With the Gentle Operator Lemma, we have
    \begin{align}
        \| T_L \rho_\alpha T_L - \rho_\alpha\|_1 \leq 2^{1-L/2} \label{eq: gentleopTL}.
    \end{align}
    We define the quantum channel $\mathcal{T}_L(\rho) = T_L \rho T_L + (1-\tr(T_L \rho))\braket{0}{0}$, which ensures
    \begin{align}
        \|  \mathcal T_L( \rho_\alpha ) - \rho_\alpha\|_1 \leq 2^{2-L/2}.
    \end{align}
    It further holds for every state $\rho$ and operator $\Pi$ that
    \begin{align}
        \tr[\Pi\mathcal T_L(\rho)^{\otimes n}] = \tr[(\mathcal T^{*}_L)^{\otimes n}(\Pi)\rho ],
    \end{align}
   where $\mathcal T^*$ is the dual of $\mathcal T$. In particular for $\rho=\rho_\alpha$ and for $\Pi$ satisfying $0\leq\Pi\leq\sum_{k=0}^{L-1}|k\rangle\langle k|$ this implies that $\Lambda:=(\mathcal T^{*})^{\otimes n}(\Pi)$ has the same probability of detecting $\rho_\alpha$ that $\Pi$ had in detecting $\mathcal T^{\otimes n}(\rho_\alpha)$. Codes for the quantum channel $\mathcal T_L\circ \mathcal N$ involving $L$-dimensional POVM elements can therefore be transformed into codes for $\mathcal N$ without any loss of performance. 
    In order to apply the bounds given in \cite[Eq. (50)-(52)]{prd}, we explicitly carve out the effect of our logarithmic dimension $L=\log(n)$. With $\tilde\rho_\alpha:=\mathcal T(\rho_\alpha)$ we have
    \begin{align}
       \tr(\Pi_\delta^n(\Tilde{\rho})\Tilde{\rho}^{\otimes n}) &\geq 1- 2\log(n)\exp(-2n\delta^2)\\
    &\geq 1 - 2^{-b'n\delta^2} \label{eq:first}\\
    2^{-n(H(\Tilde{\rho}) + c'\delta)}\Pi_\delta^n(\Tilde{\rho})&\leq \Pi_\delta^n(\Tilde{\rho})\Tilde{\rho}^{\otimes n}\Pi_\delta^n(\Tilde{\rho})\nonumber\\
    &\leq 2^{-n(H(\Tilde{\rho})-c'\delta)}\Pi_\delta^n(\Tilde{\rho})\label{eq:second}\\
      \tr(\Pi^n_\delta(\Tilde{\rho}))  &\leq 2^{n(H(\Tilde{\rho})+c'\delta)},  \label{eq:third}
    \end{align}
with adjusted constants $b', c' >0$. 
The first inequality, Eq. \eqref{eq:first}, is directly derived from \cite[Lemma 2.12]{Csiszár_Körner_2011}. The second inequality Eq. \eqref{eq:second} is derived from \cite[Property 15.1.3]{wilde2013quantum}. The third inequality Eq. \eqref{eq:third} is derived from \cite[Lemma 2.3]{Csiszár_Körner_2011} and \cite[Property 15.1.2]{wilde2013quantum}. We have an adjusted $c' = \log \log n$, following Fannes' inequality.

Let $\Pi_m\in\mathcal B(\mathcal F_L^{\otimes n})$ be a $\delta$-typical projector corresponding to message $m$ with $\tr(\Pi_m\rho_{\alpha_m}) \geq 1-\eta$. 

Each message pair $(m_1,m_2)$ is associated with a total of $M_1 + M_2$ bins: $M_1$ bins ${\mathcal{V}^{(1)}_{m_1}}, m_1 \in [M_1]$ for receiver 1 and $M_2$ bins ${\mathcal{V}^{(2)}_{m_2}}, m_2 \in [M_2]$ for receiver 2.
Each bin $\mathcal{V}^{(1)}_{m_1}$ and $\mathcal{V}^{(2)}_{m_2}$, for receiver $i$ respectively, is randomly populated. The expected bin sizes are $|\mathcal{V}^{(i)}_{m_i}| \approx e^{n\Tilde{R}_i}$. Then for every $i \in \{1,2\}$ and $m_i \in [M_i]$ it is decided to include $v$ in the bin $\mathcal{V}_{m_i}^{(i)}$ with probability $\Pr(v\in\mathcal{V}_{m_i}^{(i)}) = e^{-n (R_P - \Tilde{R}_i)}$.
The sender then transmits the message pair $(m_1,m_2)$ from a common bin $\mathcal{V}_{m_1,m_2} \in \mathcal{V}^{(1)}_{m_1} \cap \mathcal{V}^{(2)}_{m_2}$ to ensure that both receivers can successfully identify their respective messages. If the intersection is empty, a default codeword will be sent.

Each receiver performs the measurement defined by Kraus operators $T_L,|0\rangle\langle L|,|0\rangle\langle L+1|\ldots$, thereby realizing the channel $\mathcal T_L$, followed by binary hypothesis testing with POVMs
\begin{align}
\mathbf{\Pi}_{B_1^n}^{m_1} = \{\mathbf{1}-\Pi_{m_1}^{(1)}, \Pi_{m_1}^{(1)} \}, \quad  
\mathbf{\Pi}_{B_2^n}^{m_2} = \{\mathbf{1}-\Pi_{m_2}^{(2)}, \Pi_{m_2}^{(2)} \}.
\end{align}  
As mentioned, the receiver accepts the message if measurement outcome $s_{m_i} \approx 1$ and rejects it otherwise. In particular, if the outcome of the first measurement is not $L$, then the message is rejected. Let $\tilde{\Pi}_{m_i}^{(i)}:=\Pi_{m_i}^{(i)}T_L$.

    Then, we can analyse the missed identification error. This happens when the POVM $\Pi^{(i)}_{m_i}$ fails to detect message $m_i$. For receiver 1, this happens with probability
    \begin{align}
        \min_{1\leq m_1\leq M_1} \tr\left(\tilde{\Pi}^{(1)}_{m_1}\otimes_{i=1}^n\mathcal T_L(S_{N_1}(\sqrt{\tau_1}\alpha_{m_1,i})) \right) &\geq 1-\epsilon.
    \end{align}
where $\epsilon = \eta+2^{1-L/2}$. 

By equation \eqref{eq: gentleopTL}, we then get for the missed ID error
\begin{align}
e_{1,m_1}^{(1)} \leq \lambda_1 := \eta + 2^{1 - L(n)/2} = \eta + \frac{2}{\sqrt{n}}.
\end{align}
As $n\to\infty$, it follows $e_{1,m_1}^{(1)} \to 0$.
Similarly, we can show $e_{1,m_2}^{(2)} \leq \lambda_2$, such that $e_{1,m_2}^{(2)} \to 0$ as $n\to\infty$.

A false identification error occurs when receiver $i$ identifies message $ m'_i$ when actually $m_i$ was sent with $ m'_i \neq m_i$. This can happen in two scenarios: (1) When the codeword bins $\mathcal{V}_{m_i}^{(i)}$ and 
$\mathcal{V}_{m'_i}^{(i)}$, assigned to messages $m_i$ and $m'_i$ respectively, overlap. (2) When a codeword $v \in \mathcal{V}_{m_i}^{(i)}$ is incorrectly accepted by the decoder associated with $ m'_i$, i.e., the corresponding POVM $\Pi_{m'_i}^{(i)}$ produces a positive outcome even though the transmitted codeword belongs to $\mathcal{V}_{m_i}^{(i)}$.

We will first analyse the first scenario.  
\begin{align}
    \Pr(\text{overlap}) = \frac{|\mathcal{V}_{m_i}^{(i)} \cap \mathcal{V}_{m'_i}^{(i)}|}{\mathcal{V}_{m_i}^{(i)}},
\end{align}
gives the probability of overlapping bins.
Lemma \ref{lem:lemma5} gives bounds on the expected bin sizes, e.g., the overlap is bounded by $|\mathcal{V}_{m_i}^{(i)} \cap \mathcal{V}_{m'_i}^{(i)}| < 2\delta_n e^{n\Tilde{R}_i}$, with $\delta_n = e^{-n\mu/2}$ and $0 < \mu < R_P - \Tilde{R}_i$. This gives us the bound
\begin{align}
\frac{|\mathcal{V}_{m_i}^{(i)} \cap \mathcal{V}_{m'_i}^{(i)}|}{|\mathcal{V}_{m_i}^{(i)}|} < \frac{2\delta_n}{1-\delta_n} < \delta_n.
\end{align}

We have to analyse the quantum hypothesis testing error for the second scenario. 

We use the POVM $\tilde{\Pi}_{m_i'}^{(i)}:= \Pi_{m_i'}^{(i)}T_L$, which leads to the probability of falsely identifying message $m'_i$ when $m_i$ was sent is given by
\begin{align}
    \Pr(&\text{false positive}) \nonumber\\
   &= \frac{1}{|\mathcal{V}^{(i)}_{m_i} \cap (\mathcal{V}_{m'_i}^{(i)})^c|}\nonumber \\
   &\sum_{v\in \mathcal{V}^{(i)}_{m_i} \cap (\mathcal{V}_{m'_i}^{(i)})^c} \Pr(\exists v' \in \mathcal{V}^{(i)}_{m'_i}: \tr(\tilde{\Pi}_{m_i'}^{(i)}\rho_{\alpha_m}) > \eta),
\end{align}
where $\eta$ is some threshold for the quantum hypothesis test.
We apply the union bound,
\begin{align}
   \Pr(\exists v' \in \mathcal{V}^{(i)}_{m'_i}:& \tr(\tilde{\Pi}_{m_i'}^{(i)}\rho_{\alpha_m}) > \eta) \nonumber\\
   &\leq \sum_{\alpha_{m'} \in \mathcal{V}_{m'}^{(i)}} \tr(\tilde{\Pi}_{m_i'}^{(i)}\rho_{\alpha_m})\\
  & \leq 2^{-n(I(X;B_i)_\rho -  \delta)} \label{eq:typicality},
\end{align}
Note that $ I(X; B_i)_\rho = \chi(\{p(x), \rho_x^{B_i}\}) $, where $ \rho_x^{B_i} $ is the output state at receiver $i$ when the input $x \in \mathbf{A}$ is sent and $p(x)$ is the input distribution and $\chi(\{p(x), \rho_x^{B_i}\})$ is the Holevo quantity. This follows the standard identity between mutual information and Holevo information for classical-quantum states \cite[Chapter 11]{wilde2013quantum}. Eq. \eqref{eq:typicality} is derived from typicality properties. 

So, we have the total false identification error, which consists of two terms: the probability of the bins for $m_i$ and $ m'_i$ overlapping and the quantum hypothesis testing error. 
\begin{align}
    e_{2,m_i}^{(i)} \leq e^{-n\mu/2} + e^{-n(I(X;B_i)_\rho -\tilde{R}_i -\delta)}.
\end{align}
 The total false identification error decays exponentially in $n$ provided that $\Tilde{R}_i < I(X; B_i)_\rho$, which is true by the design of the code.

All error probabilities decay exponentially as $n\to\infty$, concluding the proof.
\end{proof}

\begin{proof}[Proof of Theorem~\ref{thm:main}]
    By Lemma \ref{lem:approximation-lemma}, we know that the Gaussian ensemble $G(0,E)$ can be approximated by a sequence of discrete alphabets $\mathbf{A}$ such that:
    \begin{equation}
        |\chi(\{p_x(\alpha),\mathcal N_{\tau,E}(\alpha)\}_{\alpha\in \mathbf{A}}) - g(\tau E) + g(\tau E + N)|\leq \epsilon_x,
    \end{equation}
    where $\lim_{x \to \infty} \epsilon_x = 0$.

    Applying this approximation to both receivers, we obtain the following inequalities for the achievable identification rates:
    \begin{align}
        \left| R_1 + \left( g(\tau_1 E + N_1) - g(N_1 ) \right) \right| \leq \epsilon_{x,1}, \\
        \left| R_2 + \left( g(\tau_2 E +  N_2) - g( N_2) \right) \right| \leq \epsilon_{x,2}.
    \end{align}
    Since $\lim_{x\to\infty}\epsilon_{x, i} = 0$ for $i \in \{1,2\}$, the identification rates for the discrete alphabets $\mathbf{A}_x$, $x\in\mathbb N$, converge to the identification rates for the continuous Gaussian ensemble. Thus, we conclude that:
    \begin{equation}
        \lim_{x \to \infty} \mathcal{R}_{\text{ID}, \mathbf{A}_x}^{\text{noisy}}(\mathcal{N}) = \mathcal{R}_{\text{ID}}^{\text{noisy}}(\mathcal{N}).
    \end{equation}
    This completes the proof of Theorem \ref{thm:main}.
\end{proof}

\section{Discussion}

We have demonstrated the achievable rate regions for identification over noisy bosonic broadcast channels under power constraints. The key technique used to extend results from finite-dimensional systems to infinite-dimensional ones involved first approximating the infinite sender alphabet with a discrete subset and then introducing a non-destructive energy measurement at the receiver as a mathematical tool. 

This work extends previous studies on quantum identification to noisy bosonic broadcast channels. While prior results focused on finite-dimensional quantum systems, we bridge the gap between finite-dimensional quantum information theory and practical optical communication. Our analysis shows that the achievable identification rates can be characterised using the Holevo quantity.

Furthermore, our results indicate that quantum hypothesis testing remains a valuable method for analysing identification capacity, even under infinite-dimensional constraints.

\section{Conclusion and Future Work}

We have analysed identification over noisy bosonic broadcast channels and established achievable rate regions under power constraints. Our findings demonstrate that the identification error probabilities remain bounded even when considering an infinite sender alphabet. By leveraging coherent states and quantum hypothesis testing, we provide a framework for understanding identification in noisy bosonic communication systems.

Our research also raises several open questions. While coherent states are an effective encoding strategy, investigating alternative encodings, such as squeezed states, might provide new insights. 

From a practical perspective, experimental validation of our theoretical findings will be essential for implementing identification-based communication in real-world quantum networks.

\section*{Acknowledgements}
The authors acknowledge the financial support by the Federal Ministry of Education and Research of Germany in the programme of "Souver\" an. Digital. Vernetzt.". Joint project 6G-life, project identification number: 16KISK002, and further under grant numbers 16KISQ093, 16KISR026, 16KISQ077. Further funding was received from the DFG in the Emmy-Noether program under grant number NO 1129/2-1, and by the Bavarian state government via the 6GQT and NeQuS projects. Support was provided by the Munich Center for Quantum Science and Technology (MCQST).
The research is part of the Munich Quantum Valley, supported by the Bavarian state government with funds from the Hightech Agenda Bayern Plus.

\bibliographystyle{IEEEtran}
\bibliography{references}

\end{document}